\documentclass{llncs}
\usepackage{graphicx}
\usepackage{amsmath}
\usepackage{tikz}
\title{Maximal Matching and Path Matching Counting in Polynomial Time for Graphs of Bounded Clique Width}
\author{Benjamin Hellouin de Menibus\inst{1} and Takeaki Uno\inst{2}}
\institute{ENS Lyon, France, {\tt benjamin.hellouin\_de\_menibus@ens-lyon.fr}
\and National Institute of Informatics, Tokyo 101--8430, Japan,
 {\tt uno@nii.jp}
}

\begin{document}
\maketitle
\begin{abstract}
In this paper, we provide polynomial-time algorithms for different extensions
 of the matching counting problem, namely maximal matchings, path
 matchings (linear forest) and paths, on graph classes of bounded clique-width.
For maximal matchings, we introduce matching-cover pairs to efficiently
 handle maximality in the local structure, and develop a polynomial
  time algorithm.
For path matchings, we develop a way to classify the path matchings in 
a polynomial number of equivalent classes.
Using these, we develop dynamic programing algorithms that run in 
 polynomial time of the graph size, but in exponential time of the clique-width.
In particular, we show that for a graph $G$ of $n$ vertices and clique-width
 $k$, these problems can be solved in $O(n^{f(k)})$ time where $f$ is
 exponential in $k$ or in $O(n^{g(l)})$ time where $g$ is linear or quadratic
 in $l$ if an $l$-expression for $G$ is given as input.
\end{abstract}

\section{Introduction}

Counting problems in graphs can be very difficult, i.e. $\#P$-hard in
 the general case, even for simple objects such as trees and independent sets.
Research on graph classes has been motivated by such ``hard" decision or
 optimization problems, and restricting the input to given graph classes
has led to numerous polynomial-time algorithms.
Despite this, only a few useful algorithms for counting problems exist, and these are relatively recent.\\

In this paper, we focus on maximal matching counting and path
 matching (linear forest) counting problems.
Matching counting and all extensions considered in this paper have been
 proved $\#P$-complete in the general case.
Some sparse graph classes such as planar graphs or graphs of bounded
 tree-width allow polynomial-time algorithms for perfect matching counting
 (see \cite{bib11} and \cite{bib1}); on the negative side, Valiant,
 when introducing the class $\#P$, proved that counting perfect matchings
 as well as general matchings in bipartite graphs was
 $\#P$-complete \cite{bib17,bib18}.
Valiant's proof concerning matchings has since been extended to 3-regular bipartite graphs \cite{bib8}, bipartite graphs of maximum degree 4 and bipartite planar graphs of maximum degree 6 \cite{bib16}.

The problem of counting perfect matchings in chordal and chordal bipartite graphs is also 
$\#P$-complete \cite{bib14}, but good results on independent sets \cite{bib13}
give the impression that the chordal structure could nevertheless be interesting 
regarding matching counting. This led us to focus on a related graph class, 
 the $(5,2)$-crossing-chordal graphs.
We especially make use of the bounded clique-width of this graph class.

Courcelle et al.~introduced clique-width in \cite{bib5} as a generalization
 of tree-width, and it attracted attention mainly for two reasons.
On the one hand, in a similar fashion as the tree width, putting a bound on
 the clique-width makes many difficult problems solvable in polynomial time
 (see for example \cite{bib6}). 
On the other hand, this class contains dense graphs as well as sparse graphs, 
 which makes for more general results.

Makowsky et al.~already proved as a consequence of a result in \cite{bib12}
 that matching counting on graphs of bounded clique-width is polynomial.
In this paper, we will extend this result by adapting their method to
 maximal matchings and path matchings.
Our algorithms are polynomial of the graph size, but exponential of the
 clique-width $k$, i.e., $O(n^{poly(k)})$ time.
It might be hard to develop a fixed parameter tractable algorithm such as an
 $O(c^{poly(k)}poly(n))$ time algorithm, since many graph
 algorithms, e.g. vertex coloring, have to spend $O(n^{poly(k)})$ time 
 unless FPT $\ne$ $W[1]$ \cite{bib9a}.

The existing matching counting algorithms can not be used to count 
maximal matchings directly.
The algorithms in \cite{bib12} classify matchings of local graphs according to their sizes and the colors of the endpoints, and then get information about larger graphs my merging the matchings.
However, in this way, each classified group may contain both matchings included in maximal matchings and those not included in any maximal matching.
Actually, it seems to be difficult to characterize the number of matchings included in some maximal matching, by using only their sizes and their endpoints.
In this paper, we introduce matching-cover pairs for this task. When we restrict a maximal matching to a subgraph, it can be decomposed into the matching edges belonging to the subgraph and end vertices of matching edges not included in the subgraph.
 From the maximality, the end vertices form a vertex cover of the edges of the subgraph. 
Thus, we count such pairs of matching and vertex cover according to their sizes and colors, and obtain a polynomial time algorithm for the problem.
 
For the problem of counting paths and path matchings, we have to have some way to handle 
 the connectivity of edge sets.
Actually, connectivity is not easy to handle; for example, 
checking for the existence of Hamiltonian path is equivalent to checking whether
the number of paths of length $n-1$ is larger than zero or not.
Gimenez et al. devised an algorithm based on Tutte polynomial computation 
 to count the number of forests in bounded-clique-width graphs in 
 sub-exponential time, running in $2^{O(n^c)}$ time for constant
 $c<1$ \cite{GmHlNy05}.
We use the properties of bounded-clique-width graphs so that we can 
 classify the path matchings in a polynomial number of groups of 
 equivalent path matchings, and thereby compute the number of paths
 and path matchings in polynomial time.

\section{Clique Width}

We shall introduce clique-width on undirected, non-empty labeled graphs by a construction method. Let $G_i$ be the subgraph of vertices labeled $i$ in a graph $G$. We define the singleton $S_i$ as the labeled graph with one vertex of label $i$ and no edge, and the following construction operations:

\begin{itemize}
\renewcommand{\labelitemi}{-}
\item Renaming : $\rho_{i\rightarrow j}(G)$ is $G$ where all labels $i$ are replaced by labels $j$;
\item Disjoint union : $(V_1,E_1)\oplus(V_2,E_2) = (V_1\cup V_2,E_1\cup E_2)$;
\item Edge creation : $\eta_{i,j}((V,E)) = (V,E\cup \{(v_1,v_2)\ |\ v_1 \in G_i, v_2\in G_j\})$.
\end{itemize}

The class of graphs with clique-width $\leq k$ is the smallest class containing the singletons $S_i$, closed under $\rho_{i\rightarrow j}, \oplus$ and  $\eta_{i,j}$ ($1\leq i,j \leq k$). In other words, the {\em clique-width} of a graph $G$, denoted as $cwd (G)$, is the minimal number of labels necessary to construct $G$ by using singletons and renaming, disjoint union and edge creation operations. \\

For an unlabeled graph $G$, we define its clique-width by labeling all vertices with label 1. This is necessarily the best labeling, since any labeling can be renamed to a monochromatic labeling. Note that the clique-width of a graph of order $n$ is at most $n$.

$(5,2)$-crossing-chordal graphs are known to have clique-width $\leq 3$ \cite{bib3}
 (we recall that a $(5,2)$-crossing-chordal graph is a graph where any cycle of length $\geq 5$ has a pair of crossing diagonals).
Other interesting results include: cographs are exactly the graphs with
 $cwd(G)\leq 2$, planar graphs of bounded diameter have bounded
 clique-widths, and any graph class of treewidth $\leq k$ also has a bounded
 clique-width of $\leq 3.2^{k-1}$ \cite{bib4}.
A complete review can be found in \cite{bib10}.\\

An {\em $l$-expression} is a term using
 $S_i, \rho_{i\rightarrow j}, \eta_{i,j}$ and $\oplus$ (with $i,j \leq l$) that respects the arity of each operation. It can be represented more conveniently in a tree structure, and we can inductively associate the current state of the construction to each node. If $G$ is the graph associated with the root, we say that this term is an $l$-expression for $G$, and it is a certificate that $G$ is of clique-width $\leq l$. An example is given in Fig.1.

\begin{figure}[t]
 \centering
  \begin{tikzpicture}
  [style/.style={circle,draw = black, inner sep=1pt,minimum size=3.5mm},
  small/.style={circle,draw = black, inner sep=0.5pt,minimum size=1.5mm}]
  
  \node (1) at (0,3) [style] {};
  \node (2) at (0,5) [style] {};
  \node (3) at (1,4) [style] {};
  \node (4) at (2,3) [style] {};
  \node (5) at (2,5) [style] {};
  \draw (4) -- (1) -- (2) -- (3) -- (4) -- (5) -- (2) (3) -- (5);

  \node (0) at (6,7) [style] {\scriptsize$\eta_{1,3}$};
  \node (1) at (6,6) [style] {\scriptsize{$\oplus$}};
  \node (10) at (5,5) [style] {\scriptsize$S_3$};
  \node (11) at (7,5) [style] {\scriptsize$\eta_{1,2}$};
  \node (12) at (7,4) [style] {\scriptsize{$\oplus$}};
  \node (120) at (6,3) [style] {\scriptsize{$\oplus$}};
  \node (121) at (8,3) [style] {\scriptsize$\rho_{2\to1}$};
  \node (1200) at (5,2) [style] {\scriptsize$S_2$};
  \node (1201) at (7,2) [style] {\scriptsize$S_2$};
  \node (122) at (8,2) [style] {\scriptsize$\eta_{1,2}$};
  \node (123) at (8,1) [style] {\scriptsize{$\oplus$}};
  \node (1230) at (7,0) [style] {\scriptsize$S_1$};
  \node (1231) at (9,0) [style] {\scriptsize$S_2$};
  \draw (0) -- (1) -- (10) (1) -- (11) -- (12) -- (120) -- (1200) (120) -- (1201) (12) -- (121) -- (122) -- (123) -- (1230) (123) -- (1231);

  \node (1) at (9,6.5) [small] {\tiny3};
  \node (2) at (9,7.5) [small] {\tiny2};
  \node (3) at (9.5,7) [small] {\tiny1};
  \node (4) at (10,6.5) [small] {\tiny2};
  \node (5) at (10,7.5) [small] {\tiny1};
  \draw (4) -- (1) -- (2) -- (3) -- (4) -- (5) -- (2) (3) -- (5);
  \draw [->, thick, >=stealth] (8.4,7) -- (6.7, 7);

  \node (2) at (9,5.5) [small] {\tiny2};
  \node (3) at (9.5,5) [small] {\tiny1};
  \node (4) at (10,4.5) [small] {\tiny2};
  \node (5) at (10,5.5) [small] {\tiny1};
  \draw (2) -- (3) -- (4) -- (5) -- (2) (3) -- (5);
  \draw [->, thick, >=stealth] (9,5) -- (7.6, 5);

  \node (1) at (9.5,2.75) [small] {\tiny1};
  \node (2) at (10,3.25) [small] {\tiny1};
  \draw (1) -- (2);
  \draw [->, thick, >=stealth] (9.3,3) -- (8.5, 3);

  \node (1) at (9.5,1.75) [small] {\tiny2};
  \node (2) at (10,2.25) [small] {\tiny1};
  \draw (1) -- (2);
  \draw [->, thick, >=stealth] (9.3,2) -- (8.5, 2);

  \node (1) at (4,2.75) [small] {\tiny1};
  \node (2) at (4.5,3.25) [small] {\tiny1};
  \draw (1) -- (2);
  \draw [->, thick, >=stealth] (4.7,3) -- (5.5, 3);
  \end{tikzpicture}

 \caption {Graph of clique-width 3, and a possible 3-expression tree (the last renaming operations are omitted).}
\end{figure}

Fellows et al.~proved the NP-hardness of computing the minimum clique-width for general graphs \cite{bib9}. The current best approximation is due to Oum and Seymour \cite{bib15}, who provided an algorithm in linear time that, given a graph $G$ and an integer $c$ as input, returns an $2^{3c+2}$-expression for $G$ or certifies that the graph has a clique-width larger than $c$.

This implies that we can compute in quadratic time a $2^{3k+2}$-expression for a graph of clique-width $k$ by applying this algorithm for $c=1,2\dots$. As the bound is independent of $n$, algorithms requiring expressions as input will still be in polynomial time, although the time complexity will usually be extremely poor. For $(5,2)$-crossing-chordal graphs, though, this is not a concern since it is possible to compute a 3-expression in linear time \cite{bib3}.

An $l$-expression is called {\em irredundant} if every edge-creation operation $\eta_{i,j}$  is applied to a graph where no pair of vertices in $G_i$ and $G_j$ are adjacent. Any $l$-expression can be turned into an $l$-irredundant expression in linear time \cite{bib7}. Therefore, we can assume w.l.o.g. that the input expression is irredundant.

\section{Framework of Our Algorithms}

The input of our algorithms is a graph $G$ on $n$ vertices and an $l$-expression for $G$, and the output is the number of objects (ex. matchings, paths) in $G$. The procedure works by counting these objects at each step of the construction, by using the expression tree : we start from the leaves and process a node once all its children have been processed. Finally, the value at the root of the tree is the output of the algorithm. Instead of doing it directly with the considered object, we introduce appropriate intermediate objects, and we compute tables of values at each step. 

To avoid tedious case studies, we shall assume that requesting the value of
 any vector outside of the range $\{0\dots n\}$ returns the value 0. Also, $\Delta_r(l)$ is the vector $(\delta_{i,r})_{1\leq i\leq l}$,
 and $\Delta_{r,s}(l)$ is the vector $(\delta_{i,r}\cdot\delta_{j,s})_{\substack{0\leq i\leq j\leq l\\(i,j)\not = (0,0)}}$, where $\delta_{i,j}$ is the {\em Kronecker delta}:\[\delta_{i,j} = \left\{\begin{array}{rl}1&\mbox{if }i=j\\0&\mbox{otherwise}\end{array}\right.\]We will omit the $l$ when it is obvious in context.
\section{Counting Maximal Matchings}

\begin{theorem}
 Computing the number of maximal matchings of a graph with $n$ vertices with a corresponding $l$-expression can be done in polynomial time in $n$ (but exponential w.r.t $l$).
\end{theorem} 

We cannot directly use the previous framework on maximal matchings. Indeed, consider $M$ a maximal matching of $G = \eta_{i,j}(G')$ and $M'$ the
 induced matching in $G'$: $M'$ is not necessarily maximal. However, we can keep track of the vertices of $G'$ that are covered in $M$, and those vertices must form a vertex cover of the subgraph left uncovered by $M'$. See Fig.2 for an example.\\

\begin{figure}[t]
 \centering
  \begin{tikzpicture}
  [style/.style={circle,draw = black, inner sep=1pt,minimum size=3.5mm},
   tvick/.style={circle, very thick, draw = black, inner sep=1pt,minimum size=3.5mm}]
  \node (0) at (0,6) [style] {\small3};
  \node (10) at (2,6) [style] {\small2};
  \node (11) at (2,4.5) [style] {\small2};
  \node (12) at (2,3) [style] {\small2};
  \node (13) at (2,1.5) [style] {\small2};
  \node (14) at (2,0) [style] {\small2};
  \node (21) at (4,4.5) [style] {\small1};
  \node (22) at (4,3) [style] {\small1};
  \node (23) at (4,1.5) [style] {\small1};
  \node (24) at (4,0) [style] {\small1};
  \draw (0) -- (10) -- (21) -- (11) -- (22) -- (12) -- (23) -- (14) -- (21) -- (13) -- (24) -- (12) -- (11) -- (24) -- (10) -- (23) -- (22) -- (13) -- (12) -- (21) (10) -- (22) -- (14) (23) -- (11);
  \draw [very thick] (14) -- (24) (13) -- (23) (10) -- (11) (21) -- (22);
  \draw (11) to [bend right = 20] (13);

  \node (0) at (7,6) [style] {\small3};
  \node (10) at (9,6) [style] {\small2};
  \node (11) at (9,4.5) [style] {\small2};
  \node (12) at (9,3) [style] {\small2};
  \node (13) at (9,1.5) [tvick] {\small2};
  \node (14) at (9,0) [tvick] {\small2};
  \node (21) at (11,4.5) [style] {\small1};
  \node (22) at (11,3) [style] {\small1};
  \node (23) at (11,1.5) [tvick] {\small1};
  \node (24) at (11,0) [tvick] {\small1};
  \draw (0) -- (10) (11) -- (12) -- (13) (22) -- (23);
  \draw [very thick] (10) -- (11) (21) -- (22);
  \draw (11) to [bend right = 20] (13);
  \end{tikzpicture}
\caption{Maximal matching of $\eta_{1,2}(G')$, and the corresponding matching-cover pair of $G'$.}
\end{figure}
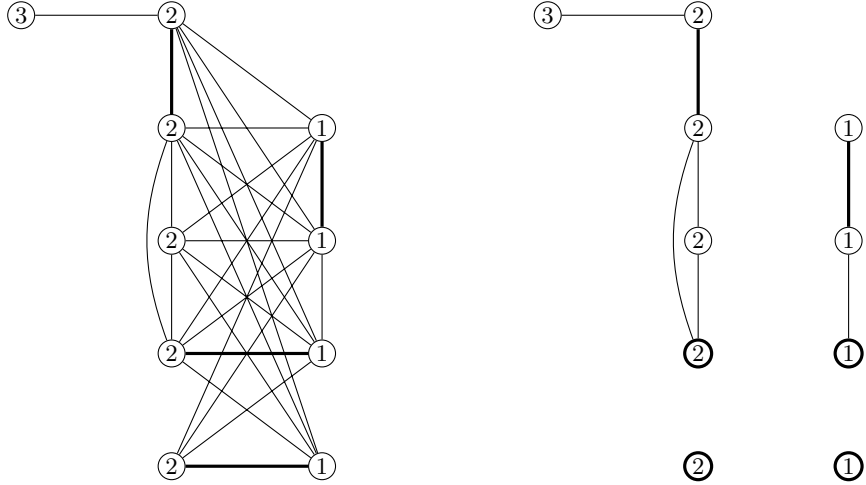

A {\em matching-cover} pair of a graph $G = (V,E)$ is a pair $(m,c)$ such that:
\begin{itemize}
 \item $m\subseteq E$ is a matching of $G$ (i.e. no vertex is covered more than once);
 \item $c\subseteq V$ is a vertex cover of the subgraph left uncovered by $m$ (i.e. every {\em edge} is covered at least once).
\end{itemize}
We show that computing the number of matching-cover pairs of a graph with $n$ vertices with a corresponding $l$-expression can be done in polynomial time in $n$.\\

 Let $M = (m_i)_{1\leq i \leq l}$ and $C= (c_i)_{1\leq i\leq l}$ be two vectors of non-negative integers. For a graph $G$, we say that a pair $(m,c)$ satisfies the condition $\varphi_{M,C}(G)$ if $m$ covers $m_i$ vertices in $G_i$ and $c$ uses $c_i$ vertices in $G_i$ for all $i$, and we denote by $mc_{M,C}(G)$ the number of pairs that satisfy $\varphi_{M,C}(G)$. Note that maximal matchings are exactly pairs with an empty cover; therefore, the number of maximal matchings of $G$ is $\sum_{k\leq n} mc_{k\cdot\Delta_1,0}(G)$.

Now we will follow the framework described above and compute $mc_{M,C}$ for all possible $M$ and $C$, at each step of the construction. We associate to each node of the tree a table of size $n^{2l}$ corresponding to the values of $mc_{M,C}$ on this graph for $M$ and $C$ ranging from $(0,..,0)$ to $(n,..,n)$. For a singleton $S_i$, we can easily see that:

\[ mc_{M,C}(S_i) =\left\{\begin{array}{rl}
1 & \mbox{if } M = 0\mbox{ and }C = 0 \mbox{ or }\Delta_i \\
0 & \mbox{otherwise}
\end{array}\right. \]

For the renaming operation $G = \rho_{i\rightarrow j}(G')$, the graph is not modified, but all vertices of label $i$ are set to label $j$. Hence, we modify the entries $i$ and $j$ accordingly. 

\[mc_{M,C}(G)=\sum_{\substack{M':(M,M')\vdash \phi_{i,j}\\C':(C,C')\vdash \phi_{i,j}}} mc_{M',C'}(G')\]

\[ \mbox{where }(X,X')\vdash\phi_{i,j} \Leftrightarrow \left(\begin{array}{l}
x_j = x'_i+x'_j\\
x_i=0\\
\forall k \not\in \{i,j\}, x_k=x'_k
\end{array}\right)\]

For the disjoint union of two graphs $G=G_1\oplus G_2$, we have a bijection between matching-cover pairs $(m,c)$ in $G$ and pairs $(m_1,c_1),(m_2, c_2)$ of matching-cover pairs in $G_1$ and $G_2$, respectively. Moreover, if $(m,c)$ satisfies $\varphi_{M,C}$, $(m_1,c_1)$ satisfies $\varphi_{M_1,C_1}$ and $(m_2, c_2)$ satisfies $\varphi_{M_2,C_2}$, we have $M=M_1+M_2$ and $C=C_1+C_2$. 

\[mc_{M,C}(G)=\sum_{\substack{M_1+M_2=M\\C_1+C_2=C}} mc_{M_1,C_1}(G_1)\cdot mc_{M_2,C_2}(G_2)\]

For the edge creation operation $G=\eta_{i,j}(G')$, we have to choose the extremities of the edges added to the matching among the vertices in the vertex cover. If $q$ is the number of new edges, we have:

\[mc_{M,C}(G)=\sum_{q=0}^n mc_{M',C'}(G')\cdot
 \left( \begin{array}{l}
 c'_i\\q
 \end{array} \right)
 \cdot
 \left( \begin{array}{l}
 c'_j\\q
 \end{array} \right)\cdot q!
\] 

\[\mbox{where } M'=M-q\Delta_i-q\Delta_j,\ C'=C+q\Delta_i+q\Delta_j\]

Once the maximal matchings of all sizes are computed, it is straightforward
 to count the number of perfect matchings and the number of minimum maximal matchings in polynomial time.
Note that counting perfect matchings can be achieved in $O(n^{2l+1})$ time by adapting the matching counting algorithm presented in \cite{bib12} in a similar fashion.\\

{\bf Complexity study:}
Obviously, there are exactly $n$ singleton operations, and each operation requires a constant amount of time. Every other operation requires one to compute $n^{2l}$ values.
 As the expression is irredundant, every edge creation operation adds at least one edge, so there are at most $n^2$ edge creation operations, processed in linear time. As a disjoint union operation has two children in the tree, and there are $n$ leaves, there are $n-1$ disjoint union operations, and they require $O(n^{2l})$ time.

For the renaming operation, consider the number of different labels at each step of the construction. This number is one for a singleton, the edge creation operation has no effect, the disjoint union is an addition in the worst case (no shared label) and the renaming operation diminishes this number by one. Therefore, there are at most $n$ renaming operations, and they are done in $O(n^4)$ time. The final sum requires $O(n^l)$ operations.\\

Therefore, the overall complexity of the algorithm is 
\[O(n)+O(n^{2l}) \cdot \left(O(n^5)+O(n^{2l+1})+ O(n^3)\right)+ O(n^l)= O(n^{4l+1})~(l\geq 2).\]
For $(5,2)$-crossing-chordal graphs, we can compute an expression of width $l=3$ in linear time and the algorithm runs in time $O(n^{13})$.

\section{Counting paths and path matchings}
A {\em path matching} (or {\em linear forest}) is a disjoint union of paths, in other words, a cycle-free set of edges such that no vertex is covered more than twice.

\begin{theorem}
Computing the number of paths $pth(G)$ and the number of path matchings
 $pm(G)$ of a graph of clique-width $\leq k$ can be done in polynomial
 time (but exponential w.r.t. $k$).
\end{theorem}

\begin{proof}
Let $K = (k_{i,j})_{\substack{0\leq i\leq j\leq l\\(i,j)\not = (0,0)}}$ be a vector of non-negative integers. We say that a path matching $P$ of $G$ satisfies the condition $\psi_K$ if:
\begin{itemize}
\renewcommand{\labelitemi}{-}
\item $\forall i>0, k_{0,i}$ vertices in $G_i$ are left uncovered by $P$;
\item $\forall (i,j),i\leq j$, $k_{i,j}$ paths in $P$ have extremities in $G_i$ and $G_j$.
\end{itemize}

We denote the number of path matchings in $G$ satisfying $\psi_K$ by $pm_K(G)$. If $i>j$, we denote $k_{i,j}= k_{j,i}$. As $K$ is of size $\frac{l(l+3)}{2}$, we compute tables of size $n^{\frac{l(l+3)}{2}}$ at each step.

For a singleton $S_i$, the only possible path matching is empty and leave the vertex uncovered.

\[\forall K, pm_K(S_i) =\left\{\begin{array}{rl}
1 & \mbox{ if }K =\Delta_{0,i}\\
0 & \mbox{otherwise}
\end{array}\right. \]

For the renaming operation  $G= \rho_{i\rightarrow{} j}(G')$, the method is the same as for maximal matchings.
\[pm_K(G) = \sum_{K':(K,K')\vdash \phi} pm_{K'}(G')\]

\[\mbox{where }(K,K')\vdash \phi \Leftrightarrow \left(\begin{array}{l}
k_{j,j}=k'_{j,j}+k'_{i,j}+k'_{i,i}\\
\forall a \not\in \{i,j\}), k_{a,j} = k'_{a,i}+k'_{a,j}\\
\forall a, k_{a,i} = 0\\
\forall a \not\in \{i,j\}, b \not\in \{0,i,j\}, k_{a,b}=k'_{a,b}
\end{array}\right)\]

For the disjoint union operation $G = G_1 \oplus G_2$, we have a bijection between path matchings $p$ in $G$ and pairs $(p_1, p_2)$ of path matchings in $G_1$  and $G_2$, respectively. Plus, if $p_1$ satisfies $\psi_{K_1}$, $p_2$  satisfies $\psi_{K_2}$ and $p$ satisfies $\psi_K$, we have $K=K_1+K_2$. 

\[pm_K(G) = \sum_{K_1+K_2=K} pm_{K_1}(G_1)\cdot pm_{K_2}(G_2)\]

Consider now the edge creation operation $G=\eta_{i,j}(G')$. We say a path matching $P$ in $G$ is an {\em extension} of a path matching $P'$ in $G'$ if $P\cap G'=P'$, so that $P=P'\cup E_{i,j}$ where $E_{i,j}$ is a subset of the edges added by the operation. Now, if we consider a path matching $P'$ in $G'$ that satisfies $\psi_{K'}$, we claim that the number of extensions of $P'$ in $G$ that satisfy $\psi_K$ depends only on $i,j,K'$ and $K$ (and not on $P'$ or $G'$), and we represent it as $N_{i,j}(K,K')$. Since every path matching of $G$ is an extension of an unique path matching of $G'$, we have:

\[ pm_K(G)=\sum_{K'} pm_{K'}(G')\cdot N_{i,j}(K', K) \]

Moreover, we can compute all the $N_{i,j}(K',K)$ beforehand in $O(n^{l(l+4)})$ time. The proof of these claims is given in the appendix.

We can then compute the number of paths $pth(G)$ and the number of path matchings $pm(G)$ with the formulas:
\[\begin{array}{rrlrl}
pth(G) =&\displaystyle \sum_{0\leq a \leq n}&pm_{K(a)}(G) &\mbox{where}&K(a)=a\cdot \Delta_{0,1}+\Delta_{1,1}\\
pm(G) =&\displaystyle \sum_{1\leq a+2b \leq n} &pm_{K(a,b)}(G) &\mbox{where}&K(a,b)=a\cdot \Delta_{0,1}+b\cdot \Delta_{1,1} 
\end{array}\]

\end{proof}

{\bf Complexity study:} A singleton operation requires constant time. Every other operation requires us to compute $n^{\frac{l(l+3)}{2}}$ values. For each value, the renaming operation in processed in linear time, the disjoint union operation in $O(n^{l^2})$ time and the edge creation operation in $O(n^{\frac{l(l+3)}{2}})$ time.

The overall complexity of the algorithm is:
\[\left\{\begin{array}{rl}
O(n^{l^2+4l})&\mbox{for }l\leq 5 \\
O(n^{\frac{3}{2}(l^2+l)+1}) & \mbox{for } l>5
\end{array}\right.\]

For $(5,2)$-crossing-chordal graphs, we can compute in linear time an expression of width $l=3$ and we have an algorithm running in $O(n^{21})$ time.

\section{Conclusion}

  These results seem to confirm the intuition that bounding clique-width is an efficient restriction on the input of $\#P$-hard problems in order to allow the use of polynomial algorithms. Notably, being able to count paths and path matchings in polynomial time is interesting because connected structures are usually very difficult to count. In that sense, the next logical step was to study the tree (or, equivalently, forest) counting problem. However, our attempts to do so by using a method similar to the one we used in the paper, only produced algorithms running in exponential time. Our feeling is that the tree counting problem remains $\#P$-complete for graphs of bounded clique-width, as this intuitive method keeps giving bad results. It remains an open problem for now.

\begin{center} {\huge{\textbf{Appendix.}}} \end{center}~\\

We now prove the case of Thm.2 we have omitted. Let $G$ and $G'$ be two labeled graphs such that $G=\eta_{i,j}(G')$ (for some $i<j$) and $P'$ a path matching of $G'$ satisfying $\psi_{K'}$ for some $K'$. For any $K$, we want to compute the number of extensions of $P'$ in $G$ satisfying $\psi_K$.\\

\textbf{Definitions.} For any path matching satisfying $\psi_{K}$, a path with two extremities $x\in G_i$ and $y\in G_j$ is called an {\em $(i,j)$-path}, and $x$ and $y$ are called {\em partners}. We denote by $V_a(b)$ the vertices of $G_a$ whose partner is in $G_b$, and by $V_a(0)$ the uncovered vertices in $G_a$. We also note $v_{a,b}=\# V_a(b)$, which means that $v_{a,b} = k_{a,b}$, except for $v_{a,a}=2k_{a,a}$ (note that $v_{a,b}$ depend only on $K$). An edge which extremities are in $V_i(a)$ and in $V_j(b)$, respectively, is called an {\em $(a,b)$-edge}.\\

We use a dynamic programming technique to build all possible extensions of $P'$ by considering each vertex of $G'_i$ one by one in the order $V_i(j),V_i(0),..,V_i(l)$ (the reason for this order will be explained later). If $X=(x_0,..,x_l)$ is a vector of non-negative integers and $P_1$ a path matching in $G_1$ that satisfies $\psi_{K_1}$, $T_{i,j}(G_1,P_1,K_2,X)$ stands for the number of extensions of $P_1$ in $\eta_{i,j}(G_1)$ that satisfy $\psi_{K_2}$ and that uses only the $x_k$ last vertices of $V_i(k)$ for every $k$.\\

At each step of the computation, the equations show that knowing $G_1$ and $P_1$ is not necessary as long as $\psi_{K_1}$ is satisfied: this proves our first claim, and we write $T_{i,j}(K_1,K_2,X)$ for $T_{i,j}(G_1,P_1,K_2,X)$. Also, since $i,j$ and $K_2$ are not modified during the computation, we write $T(K_1, X)$ for $T_{i,j}(K_1,K_2,X)$.\\

We now detail the different steps by increasing difficulty (instead of the actual order of the algorithm). First, assume that $x_j = x_0 =..=x_{k-1} = 0$ and $x_k \not = 0$ (for some $k \not = i)$. We consider the first vertex in $V_i(k)$ that has not been considered yet in the computation. We have only two possibilities:
\begin{itemize}
\renewcommand{\labelitemi}{-}
\item No new edge adjacent to this vertex is added to the path matching.
\item One new $(k,a)$-edge (possibly $a=0$) is added to the path matching: we have $v_{j,a}$ choices for the edge. A $(i,k)$-path and a $(j,a)$-path are transformed into a $(k,a)$-path.
\end{itemize}

In each case, the value of the current $K$ is updated accordingly and the vertex is deleted from $X$. Next to each term is the set that contains the other extremity of the edge being considered.

\[\begin{array}{rlc}
T(K",X) = &T(K",X - \Delta_k)&\emptyset\\
&+ \displaystyle\sum_{1\leq a \leq l} v_{j,a}.T(K" - \Delta_{i,k} - \Delta_{j,a}+\Delta_{k,a}, X - \Delta_k)&V_j(a)\\
&+v_{0,j}.T(K"-\Delta_{i,k} -\Delta_{0,j} + \Delta_{j,k}, X-\Delta_k)&V_j(0)
\end{array}\]

Note that if a $(k,i)$-edge is added, the partner of another vertex of $V_i(j)$ is also modified: this is why $V_i(j)$ is considered first in the computation, so that it does not appear in $X$ anymore at this step. This remark holds for all the other cases except for $k= j$.\\

Now, we consider the first step ($k=j$). The situation is similar, but the vertex cannot be linked to its own partner when $k'=i$. Note that adding a $(j,i)$-edge changes the partner of another vertex of $V_i(j)$, but the new partner is still in $G_j$, so doing this brings no modification to $X$.

\[\begin{array}{rlc}
T(K" ,X) =&  T(K", X - \Delta_j)&\emptyset\\
 &+ \displaystyle\sum_{\substack{1\leq a\leq l\\ a\not = i}} v_{j,a} T(K"- \Delta_{i,j}, X - \Delta_j)&V_j(a)\\
 &+ (v_{i,j} - 1) T (K" - \Delta_{i,j}, X  -\Delta_j)&V_j(i)\\
 &+  v_{0,j}.T(K"- \Delta_{i,j} - \Delta_{0,j} + \Delta_{j,j}, X - \Delta_j)&V_j(0)
\end{array}\]

For the uncovered vertices ($k=0$), up to two edges can be added to the path matching. The possibilities are:

\begin{itemize}
\renewcommand{\labelitemi}{-}
\item  No new edge adjacent to this vertex is added to the matching.
\item  One new $(k,k')$-edge is added to the matching: we have $v_{j,k'}$ choices for the edge. An uncovered vertex and a $(j,k')$-path are transformed into a $(x,k')$-path.
\item Two new $(k,k')$ and $(k,k")$-edges are added to the matching: we have $v_{j,k'}. v_{j,k"}$ choices for the two edges (only half of those when $k'= k"$). An uncovered vertex, a $(j,k')$-path and a $(j,k")$-path are transformed into a $(k',k")$-path.
\end{itemize}
\[\begin{array}{rlc} T (K", X) = & T(K", X - \Delta_0)&\emptyset\\
&+ \displaystyle\sum_{1\leq a \leq l} v_{j,a}.T(K"-\Delta_{0,i}- \Delta_{j,a}+\Delta_{i,a},X-\Delta_0)&V_j(a)\\
&+\ v_{0,j}.T(K"-\Delta_{0, i}-\Delta_{0,j} + \Delta_{i,j}, X-\Delta_0)&V_j(0)\\
&+ \displaystyle\sum_{1\leq a<b\leq l} v_{j,a}.v_{j,b}.T(K"-\Delta_{0,i}-\Delta_{j,a}-\Delta_{j,b}+\Delta_{a,b},X-\Delta_0)&V_j(a)\ |\ V_j(b)\\
&+ \displaystyle\sum_{1\leq a \leq l} v_{j,a}.v_{0,j}.T(K"-\Delta_{0,i}-\Delta_{0,j},X-\Delta_0)&V_j(a)\ |\ V_j(0)\\
&+ \displaystyle\sum_{\substack{1\leq a \leq l \\ a \not =j}}\frac{v_{j,a}.(v_{j,a}-1)}{2}.T(K"-\Delta_{0,i}-2\Delta_{j,a}+\Delta_{a,a}, X -\Delta_0)&V_j(a)\ |\ V_j(a)\\
&+\ \displaystyle\frac{v_{0,j}.(v_{0,j}-1)}{2}.T (K"-\Delta_{0,i}-2\Delta_{0,j}+\Delta_{j,j},X-\Delta_0)&V_j(0)\ |\ V_j(0)\\
&+\ \displaystyle \frac{v_{j,j}.(v_{j,j}-2)}{2}.T (K" - \Delta_{0,i}-\Delta_{j,j}, X-\Delta_0)&V_j(j)\ |\ V_j(j)
\end{array}\]

For $k=i$ , we consider the two extremities of the $(i,i)$-path at the same time. Therefore, this situation is similar to the previous one, except that we have to choose one of the extremities in each case. There are twice as many possibilities as in the previous case.

\[\begin{array}{rlc}
T(K", X) = & T(K", X - 2\Delta_i)&\emptyset\\
&+\displaystyle \sum_{1\leq a \leq l} 2v_{j,a}.T(K"-\Delta_{i,i}- \Delta_{j,a}+\Delta_{i,a},X-2\Delta_i)&V_j(a)\\
&+\ 2v_{0,j}.T(K"-\Delta_{i,i}-\Delta_{0,j} + \Delta_{i,j},X-2\Delta_i)&V_j(0)\\
&+\displaystyle \sum_{1\leq a<b\leq l} 2v_{j,a}.v_{j,b}.T(K"-\Delta_{i,i}-\Delta_{j,a}-\Delta_{j,b}+\Delta_{a,b},X-2\Delta_i)&V_j(a)\ |\ V_j(b)\\
&+\displaystyle \sum_{1\leq a \leq l} 2v_{j,a}.v_{0,j}.T(K"-\Delta_{i,i}-\Delta_{0,j},X-2\Delta_i)&V_j(a)\ |\ V_j(0)\\
&+\displaystyle \sum_{\substack{1\leq a\leq l\\ a\not =j}} v_{j,a}.(v_{j,a}-1).T(K"-\Delta_{i,i}-2\Delta_{j,a}+\Delta_{a,a}, X -2\Delta_i)&V_j(a)\ |\ V_j(a)\\
&+\ v_{0,j}.(v_{0,j}-1).T (K"-\Delta_{i,i}-2\Delta_{0,j}+\Delta_{j,j},X- 2\Delta_i)&V_j(0)\ |\ V_j(0)\\
&+\ v_{j,j}.(v_{j,j}-2).T (K" - \Delta_{i,i}-\Delta_{j,j}, X-2\Delta_i)&V_j(j)\ |\ V_j(j)
\end{array}\]

Now, if we have $\forall k,x_k  = 0$, then all the vertices have been considered and:
\[T_{i,j}(G_1, K_1, K_2,0) = \left\{\begin{array}{rl}
1 & \text{if } K_1 = K_2\\
0 & \text{otherwise}
\end{array}\right.\]

The table of all possible $T_{i,j}(K_1, K_2, X)$ is of size $l^2.n^{l(l+4)}$. Using the previous equations, we can compute the table by increasing $X$ in $O(n^{l(l+4)})$ operations (individual equations are independent of $n$). We now have $N_{i,j}(K',K) = T_{i,j}(K',K,X)$ where $\forall k\neq i,x_k = k'_{i,k}$ and $x_i = 2k'_{i,i}$.
\end{document}